\documentclass[publicdomain,creativecommons]{eptcs}
\providecommand{\event}{FICS 2013} % Name of the event you are submitting to
\usepackage{breakurl}             % Not needed if you use pdflatex only.

\usepackage{amsmath,amssymb,amsthm,stmaryrd}

\newtheorem{lemma}{Lemma}[section]
\newtheorem{theorem}[lemma]{Theorem}
\newtheorem{proposition}[lemma]{Proposition}
\newtheorem{corollary}[lemma]{Corollary}

\newtheorem{italicdefinition}[lemma]{Definition}
\newenvironment{definition}{\begin{italicdefinition}\em}{\end{italicdefinition}}

\newtheorem{italicremark}[lemma]{Remark}
\newenvironment{remark}{\begin{italicremark}\em}{\end{italicremark}}
\newtheorem{italicexample}[lemma]{Example}

\newtheorem{italicquestion}[lemma]{Question}

\newcommand{\lts}{ \mathcal{L}}

\newcommand{\lfp}{\operatorname{lfp}}
\newcommand{\gfp}{\operatorname{gfp}}

\newcommand{\sem}[1] {  \llbracket #1 \rrbracket  }  
\newcommand{\gsem}[1]{\llparenthesis{\,#1\,}\rrparenthesis}

\newcommand{\LTL}{\ensuremath{\textbf{LTL}}}
\newcommand{\CTL}{\ensuremath{\textbf{CTL}}}
\newcommand{\CTLstar}{\ensuremath{\textbf{CTL}^*}}
\newcommand{\PCTL}{\ensuremath{\textbf{PCTL}}}
\newcommand{\PCTLstar}{\ensuremath{\textbf{PCTL}^*}}

\newcommand{\Lukmu}{\ensuremath{\text{\L}\mu}}
\newcommand{\qLmu}{\ensuremath{\text{qL}\mu}}

\newcommand{\strongand}{\odot}
\newcommand{\strongor}{\oplus}
\newcommand{\weakand}{\sqcap}
\newcommand{\weakor}{\sqcup}

\newcommand{\Mu}[2]{\mu #1.\, #2}
\newcommand{\Nu}[2]{\nu #1.\, #2}
\newcommand{\negate}[1]{\overline{#1}}

\newcommand{\GLE}[2]{#1 \vdash #2}

\newcommand{\Implies}{\rightarrow}
\newcommand{\dominates}{\vartriangleright}

\providecommand{\OO}[1]{\mathop{\mathrm{O}}\bigl(#1\bigr)}

\title{{\L}ukasiewicz $\mu$-calculus}
\author{Matteo Mio 
\institute{CWI, Amsterdam (NL)}
\email{miomatteo@gmail.com}
\and
Alex Simpson
\institute{LFCS, School of Informatics \\ University of Edinburgh}
\email{Alex.Simpson@ed.ac.uk}
}
\def\titlerunning{{\L}ukasiewicz $\mu$-calculus}
\def\authorrunning{Matteo Mio \& Alex Simpson}
\begin{document}
\maketitle

\begin{abstract}
The paper explores properties of \emph{{\L}ukasiewicz $\mu$-calculus}, a version of the quantitative/probabilistic modal $\mu$-calculus containing both weak and strong conjunctions and disjunctions from {\L}ukasiewicz (fuzzy) logic. We show that this logic encodes the well-known probabilistic temporal logic  \PCTL. And we give a model-checking algorithm for computing the rational denotational value of a formula at any state in a finite rational probabilistic 
nondeterministic transition system.
\end{abstract}

\section{Introduction}

Among logics for expressing properties of nondeterministic (including concurrent) processes, represented as transition systems, Kozen's modal $\mu$-calculus~\cite{Kozen83} plays a fundamental r\^{o}le. It subsumes other temporal logics of processes, such as $\LTL$, $\CTL$ and $\CTLstar$. It does not distinguish bisimilar processes, but separates (finite) non-bisimilar ones. 
More generally, % bisimilarity is central to 
by a remarkable result of Janin and Walukiewicz~\cite{JW96},
it is exactly as expressive as 
the bisimulation-invariant fragment of monadic second-order logic. Furthermore, there is an intimate connection with parity games, which offers an intuitive reading of fixed-points, and underpins the existing technology for  model-checking $\mu$-calculus properties. % on finite processes.

For many purposes, it is useful to add probability to the computational model, leading to probabilistic nondeterministic transition systems, cf.~\cite{S95}.
%Historically, three fundamentally different  approaches have been followed in developing analogues of the modal $\mu$-calculus in this setting. In two of them, \cite{CPN99,DvG2010}, 
Among the different approaches that have been followed to 
developing analogues of the modal $\mu$-calculus in this setting, the most significant is that introduced independently by Huth and  Kwiatkowska~\cite{HM96} and
by Morgan and McIver~\cite{MM97}, under which a \emph{quantitative} interpretation is given, with formulas denoting values in $[0,1]$.
%which can loosely be understood  as probabilities.
This quantitative setting permits several variations. In particular, three different quantitative extensions of conjunction from booleans to $[0,1]$ (with  $0$ as false and $1$ as true) arise naturally~\cite{HM96}: 
minimum, $\min(x,y)$; multiplication, $xy$; and the 
strong conjunction (a.k.a.\ {\L}ukasiewicz t-norm)
from {\L}ukasiewicz fuzzy logic, $\max( x+y-1, \, 0)$. 
In each case, there is a dual operator 
giving a corresponding extension of disjunction: maximum, $\max(x,y)$; comultiplication, 
$x+y-xy$; and {\L}ukasiewicz strong disjunction,  $\min(x+y,\, 1)$. The choice of $\min$ and $\max$ for conjunction and disjunction  is particularly natural, since the corresponding $\mu$-calculus, called $\qLmu$ in \cite{MM07},
has an interpretation in terms of 2-player \emph{stochastic} parity games, which extends the usual parity-game interpretation of the ordinary modal $\mu$-calculus. 
This allows the real number denoted by a formula to be understood as the \emph{value} of the associated game~\cite{MM07,MIO2012a}.

The present paper contributes to a programme of ongoing research, one of whose
overall aims is to investigate the extent to which quantitative $\mu$-calculi
play as fundamental a r\^{o}le in the probabilistic setting as that
of Kozen's $\mu$-calculus in the nondeterministic setting. 
The logic $\qLmu$, with  min/max as conjunction/disjunction, is insufficiently expressive. For example, it cannot 
encode the standard probabilistic
temporal logic  \PCTL\ of~\cite{BA1995}.
Nevertheless, richer calculi can be obtained by augmenting $\qLmu$ with the 
other alternatives for conjunction/disjunction, to be used in combination with $\max$ and $\min$. Such extensions  were investigated by the first author in~\cite{MIO2012b,MioThesis}, where 
% it is shown that they are strictly more expressive than $\qLmu$ and that  
the game-theoretic interpretation was generalized to accommodate the new operations.

In this paper, we focus on a calculus containing two
different interpretations of conjunction and disjunction:
$\min$ and $\max$ (written as 
$\weakand$ and $\weakor$) and the {\L}ukasiewicz
operations (written as $\strongand$ and $\strongor$). In addition, as
is natural in the quantitative setting, we include a basic operation for
multiplying the value of a formula by a rational constant in $[0,1]$.
Since these operations are all familiar from 
{\L}ukasiewicz fuzzy logic (see, e.g.,~\cite{HJ98}), 
we call the resulting
 logic \emph{{\L}ukasiewicz $\mu$-calculus} ($\Lukmu$).

As our first contribution, we show that 
the standard probabilistic
temporal logic \PCTL~\cite{BA1995} can be encoded in $\Lukmu$.
A similar translation was originally given in 
the first author's PhD thesis~\cite{MioThesis}, 
%a translation of the \emph{qualititative} fragment of $\PCTL$ is given into the extension of $\qLmu$ with multiplication and comultiplication, and full 
where $\PCTL$ was translated into a quantitative 
$\mu$-calculus containing all three 
pairs of quantitative conjunction/disjunction operations in combination.
%further extension obtained by
%adding the {\L}ukasiewicz operations as a third 
Here, we streamline the treatment by
implementing the observation that the (co)multiplication
operations are not required once the {\L}ukasiewicz operations are in place.
In fact, all that is needed is the encodability of certain \emph{threshold modalities}, see Remark \ref{remark_fragment_plmu} below.

An advantage of the {\L}ukasiewicz $\mu$-calculus considered in the present paper is that it
enjoys the property that the value of a formula in a finite rational model is 
rational, a property which does not hold when the (co)multiplication operations are included in the logic. 
As our second contribution, we exploit this property by giving a 
 (quantitative) model-checking algorithm
%\footnote{Although the algorithm,  which computes a rational number, does not solve a decision problem, the use of the adjective ``model-checking'' appears to be established in the literature on quantitative temporal logics. See, e.g., \cite{HM96}.} 
that computes the 
value of a $\Lukmu$
formula at a state in a finite rational probabilistic nondeterministic 
transition system. The algorithm adapts the %standard 
approximation-based approach to nested fixed-point calculation to our quantitative calculus.
%Although inefficient, it has the advantage of conceptual simplicity, and we believe it provides a starting point for the development of 
%future improved algorithms. 

One could  combine our two contributions and obtain a new model-checking algorithm for $\PCTL$. But this is not advisable since the complexity bounds we obtain for  model-checking $\Lukmu$ are abysmal. 
The positive messages of this paper are rather that $\PCTL$ fits into the conceptually appealing framework of quantitative $\mu$-calculi, and that this framework is itself algorithmically approachable.

\section{Technical background}

\begin{definition}
Given a set $S$ we denote with $\mathcal{D}(S)$ the set of \emph{(discrete) probability distributions} on $S$ defined as $\mathcal{D}(S)\!=\!\{Êd:S\rightarrow[0,1] \ | \ \displaystyle\sum_{s\in S} d(s) = 1\}$. We say that $d\in\mathcal{D}(S)$ is \emph{rational} if $d(s)$ is a rational number, for all $s\in S$.
\end{definition}
\begin{definition}
A \emph{probabilistic nondeterministic transition system} (PNTS) is a pair $(S, \rightarrow)$ where $S$ is a set of states and $\rightarrow\ \subseteq S\times \mathcal{D}(S)$ is the \emph{accessibility} relation. We write $s\not \rightarrow$ if $\{ d \ | \ s\rightarrow d\}=\emptyset$. A PNTS $(S,\rightarrow)$ is \emph{finite rational} if $S$ is finite and  $\bigcup_{s\in S}\{ d \ | \ s\rightarrow d\}$ is a finite set of rational probability distributions.
\end{definition}

%Reference to {\L}ukasiewicz (fuzzy) logic \cite{HJ98}.

We now introduce the novel logic $\Lukmu$ which extends the probabilistic (or quantitative) modal $\mu$-calculus (qL$\mu$) of \cite{HM96,MM97,MM07,AM04}.

\begin{definition}
The logic $\Lukmu$ is generated by the following grammar:
\[
\phi \, ::= \,   X \mid P \mid \negate{P} \mid 
q \, \phi \mid  
% \convsum{\lambda}{\phi}{\phi} \mid 
 \phi \weakor \phi \mid \phi \weakand \phi \mid \phi \strongor \phi \mid \phi \strongand \phi  \mid \ \Diamond\phi \mid \ \Box \phi \mid  \Mu{X}{\phi} \mid \Nu{X}{\phi} \enspace ,\]
where $q$ ranges over rationals in $[0,1]$, $X$ over a countable set $\texttt{Var}$ of variables and $P$ over a set $\texttt{Prop}$ of propositional letters which come paired with associated complements $\negate{P}$. As a convention we denote with $\underline{1}$ the formula $\nu X.X$ and with $\underline{q}$ the formula $q\, \underline{1}$.
\end{definition}

Thus, $\Lukmu$ extends the syntax of the probabilistic modal $\mu$-calculus by the new pair of connectives ($\odot$, $\oplus$), which we refer to as \emph{\L ukasiewicz conjunction} and \emph{disjunction}, respectively, and a form of \emph{scalar multiplication} ($q\, \phi$) by rationals numbers in $[0,1]$. For mild convenience in the encoding of $\PCTL$ below, we consider a version with unlabelled modalities and propositional letters. However, the approach of this paper easily adapts to a labeled version of $\Lukmu$.

Formulas are interpreted over PNTS's as we now describe.

\begin{definition}
Given a PNTS $(S,\rightarrow)$, an \emph{interpretation} for the variables and propositional letters is a function $\rho:(\texttt{Var}\uplus\texttt{Prop}) \rightarrow (S\rightarrow[0,1])$ such that $\rho(\negate{P})(x)= 1- \rho(P)(x)$. Given a function $f:S\rightarrow [0,1]$ and $X\in\texttt{Var}$ we define the interpretation  $\rho[f/X]$  as $\rho[f/X](X)=f$ and $\rho[f/X](Y)=\rho(Y)$, for $X\neq Y$.
\end{definition}

\begin{definition}
\label{definition:semantics}
The semantics of a $\Lukmu$ formula $\phi$ interpreted over $(S,\rightarrow)$ with interpretation $\rho$ is a function $\sem{\phi}_\rho: S\rightarrow[0,1]$ defined inductively on the structure of $\phi$ as follows:

\begin{center}
\begin{tabular}{ l l l}
$\sem{X}_{\rho}=\rho(X)$    & & $\sem{q\, \phi}_{\rho}(x)= q \cdot \sem{\phi}_\rho (x)$\\
$\sem{P}_{\rho}=\rho(P)$ & & $\sem{\negate{P}}_{\rho}= 1-\rho(P)$ \\
%$\sem{q\, \phi}_{\rho}(x)= q \cdot \sem{\phi}_\rho (x)$ & & 
$\sem{\phi \sqcup \psi}_\rho(x)= \max \{ \sem{\phi}_{\rho}(x)Ê  ,  \sem{\psi}_{\rho}(x)  \}$ & & 
$\sem{\phi \sqcap \psi}_\rho(x)= \min \{ \sem{\phi}_{\rho}(x)Ê  ,  \sem{\psi}_{\rho}(x)  \}$ \\ 
$\sem{\phi \oplus \psi}_\rho(x)= \min \{ 1, \sem{\phi}_{\rho}(x)Ê  +  \sem{\psi}_{\rho}(x)  \}$ &&
$\sem{\phi \odot \psi}_\rho(x)= \max \{ 0, \sem{\phi}_{\rho}(x)Ê +   \sem{\psi}_{\rho}(x) -1  \}$ \\ 
$\sem{\Diamond \phi}_{\rho}(x)= \displaystyle \bigsqcup_{x\rightarrow d}\Big( \sum_{y\in X} d(y)\sem{\phi}_\rho(y)\Big)$&&
$\sem{\Box \phi}_{\rho}(x)= \displaystyle \bigsqcap_{x\rightarrow d}\Big(\sum_{y\in X} d(y)\sem{\phi}_\rho(y) \Big)$\\
$\sem{\mu X.\phi} = \lfp \big(  f \mapsto  \sem{\phi}_{\rho[f/X]}\big)$ & & 
$\sem{\mu X.\phi} = \gfp \big(  f \mapsto \sem{\phi}_{\rho[f/X]}\big)$\\
\end{tabular}
\end{center}
It is straightforward to verify that the interpretation of every operator is monotone, thus the existence of least and greatest points in the last two clauses is guaranteed by the the Knaster-Tarski theorem. 
\end{definition}

As customary in fixed-point logics, we presented the logic $\Lukmu$ in positive normal form. A negation operation $\texttt{dual}(\phi)$ can be defined on \emph{closed} formulas by replacing every connective with its dual and $(q\, \phi)$ with $( (1-q)\, \phi)$. It is simple to verify that $\sem{\texttt{dual}(\phi)}_\rho(x)= 1- \sem{\phi}_\rho(x)$. 

% \subsection{The logic $\PCTL$}

Next, we introduce the syntax and the semantics of the logic $\PCTL$ of \cite{BA1995}. We refer to \cite{BaierKatoenBook} for an extensive presentation of this logic. 

The notions of  \emph{paths}, \emph{schedulers} and \emph{Markov runs} in a PNTS are at the basis of the logic $\PCTL$. 
\begin{definition}\label{run_PLTS}
For a given PNTS $\lts=(S,\rightarrow)$ the binary relation $\leadsto_\lts \ \subseteq S\times S $ is defined as follows: $\leadsto_\lts= \{ (s,t) \ | \ \exists d. (s \rightarrow d \ \wedge \ d(t)>0)\}$. Note that $s\not \rightarrow$ if and only if $s\not\leadsto$.  We refer to $(S,\leadsto)$ as the \emph{graph underlying $\lts$}.
\end{definition}

\begin{definition}
A \emph{path} in a PNTS $\lts=(S,\rightarrow)$ is an ordinary path in the graph $(S,\leadsto)$, i.e., a finite or infinite sequence $\{s_i\}_{i\in I}$ of states such that $s_i\leadsto s_{i+1}$, for all $i+1\in I$. We say that a path is  \emph{maximal} if either it is infinite or it is finite and its last entry is a state $s_n$ without successors, i.e., such that $s_n\not \leadsto$. We denote with $\textnormal{P}(\lts)$ the set of all maximal paths in $\lts$. The set $\textnormal{P}(\lts)$  is endowed with the topology generated by the basic open sets $U_{\vec{s}}=\{  \vec{r}\ | \  \vec{s}\sqsubseteq \vec{r} \}$ where $\vec{s}$ is a finite sequence of states and $\sqsubseteq$ denotes the prefix relation on sequences. The space $\textnormal{P}(\lts)$ is always $0$-dimensional, i.e., the basic sets $U_{\vec{s}}$ are both open and closed and thus form a Boolean algebra. We denote with $\textnormal{P}(s)$ the open set $U_{\{s\}}$ of all maximal paths having $s$ as first state.
\end{definition}

\begin{definition}\label{scheduler_def}
A \emph{scheduler} in a PNTS $(S,\rightarrow)$ is a partial function $\sigma$ from non-empty finite sequences $s_0.\dots s_n$ of states to probability distributions $d\in\mathcal{D}(S)$ such that  $\sigma(s_0.\dots s_n)$ is not defined if and only if $s_n \not\rightarrow$ and, if $\sigma$ is defined at $s_0 . \dots s_n$ with $\sigma(s_0.\dots s_n)=d$, then $s_n\rightarrow d$ holds. A pair $(s,\sigma)$ is called a \emph{Markov run} in $\mathcal{L}$ and denoted by $M^s_\sigma$. It is clear that each Markov run $M^s_\sigma$ can be identified with a (generally) infinite Markov chain (having a tree structure) whose vertices are finite sequences of states and having $\{s\}$ as root.
\end{definition}

Markov runs are useful as they naturally induce probability measures on the space $\textnormal{P}(\mathcal{L})$.

\begin{definition}
Let $\mathcal{L}=(S,\rightarrow)$ be a PNTS and $M^s_\sigma$ a Markov run. We define the measure $m^s_\sigma$ on $\textnormal{P}(\mathcal{L})$ as the unique (by Carath{\'e}odory extension theorem) measure specified by the following assignment of basic open sets:
\begin{center}
$\displaystyle m^s_\sigma\big( \textnormal{U}_{s_0.\dots s_n} \big) = \prod^{n-1}_{i=0} d_i (s_{i+1})$
\end{center} 
 where $d_i=\sigma(s_0.\dots s_i)$ and $\prod\emptyset = 1$. It is simple to verify that $m^s_\sigma$ is a probability measure, i.e., $m^s_\sigma(\textnormal{P}(\mathcal{L}))=1$. We refer to $m^s_\sigma$ as the probability measure on $\textnormal{P}(\mathcal{L})$ induced by the Markov run $M^s_\sigma$.
\end{definition}

We are now ready to specify the syntax and semantics of  $\PCTL$.
\begin{definition}
Let the letter $P$ range over a countable set of propositional symbols $\texttt{Prop}$. The class of  $\PCTL$ \emph{state-formulas} $\phi$ is generated by the following two-sorted grammar:
\begin{center}
$ \phi::= \  \textnormal{true} \ | \ P \ | \ \neg \phi \ | \ \phi \vee \phi \ |  \ \exists  \psi\ |  \ \forall  \psi\ | \ \mathbb{P}^\exists_{\rtimes q}\psi \ | \ \ \mathbb{P}^\forall_{\rtimes q}\psi$
\end{center}
with $q\in\mathbb{Q}\cap[0,1]$ and  $\rtimes\in\{ >, \geq\}$, where \emph{path-formulas} $\psi$ are generated by the simple grammar: $\psi::= \  \circ \phi \ | \ \phi_1 \mathcal{U}\phi_2$. Adopting standard terminology, we refer to the connectives $\circ$ and $\mathcal{U}$ as the \emph{next} and  \emph{until}  operators, respectively.
\end{definition}
%Note that the expected state-formulas operators $\forall \psi$, $\mathbb{P}_{\leq q}\psi$ and $\mathbb{P}_{z q}\psi$ are definable as $\not \exists \neg \psi$, 

\begin{definition}
Given a PNTS $(S,\rightarrow)$, a \emph{$\PCTL$-interpretation} for the propositional letters is a function $\rho:\texttt{Prop}\rightarrow 2^S$, where $2^S$ denotes the powerset of $S$. 
\end{definition}

\begin{definition}
Given a PNTS $(S,\rightarrow)$ and a  $\PCTL$-interpretation $\rho$ for the propositional letters, the semantics $\gsem{\phi}_\rho$ of a  $\PCTL$ state-formula $\phi$ is a subset of $S$ (i.e., $\gsem{\phi}_\rho:S\rightarrow\{0,1\}$) defined by induction on the structure of $\phi$ as follows:
\begin{itemize}
\item $\gsem{ \textnormal{true} }_\rho= S$,
$\gsem{P}_\rho = \rho(P)$,
$\gsem{\phi_1 \vee \phi_2}_{\rho}= \gsem{\phi_1}_\rho \cup \gsem{\phi_2}_{\rho}$,
$\gsem{\neg \phi }_{\rho}= S\setminus \gsem{\phi}_{\rho}$,
 \item $\gsem{\exists \psi }_{\rho}(s) = 1$ if and only there exists $\vec{s}\in \textnormal{P}(s)$ such that that $\vec{s}\in \sem{\psi}$
  \item $\gsem{\forall \psi }_{\rho}(s) = 1$ if and only forall $\vec{s}\in \textnormal{P}(s)$ it holds that $\vec{s}\in \gsem{\psi}_\rho(\vec{s})$
 \item $\gsem{\mathbb{P}^\exists_{\rtimes q} \psi }_{\rho}(s) = 1$ if and only $\big(\bigsqcup_{\sigma} m^s_\sigma (\gsem{\psi}_\rho)\big) \rtimes q$
  \item $\gsem{\mathbb{P}^\forall_{\rtimes q} \psi }_{\rho}(s) = 1$ if and only $\big(\bigsqcap_{\sigma} m^s_\sigma (\gsem{\psi}_\rho)\big) \rtimes q$
\end{itemize}
where $\sigma$ ranges over schedulers and the semantics $\gsem{\psi}_\rho$ of path formulas, defined as a subset of $\textnormal{P}(\mathcal{L})$ (i.e., as a map $\gsem{\psi}_\rho: \textnormal{P}(\mathcal{L})\rightarrow\{0,1\}$) is defined as:
\begin{itemize}
\item $\gsem{\circ \phi}_\rho ( \vec{s}) = 1$ if and only if $| \vec{s} | \geq 2$ (i.e., $\vec{s}=s_0.s_1.\dots$) and $s_1\in \gsem{\phi}_\rho$,
\item $\gsem{\phi_1 \mathcal{U} \phi_2}_\rho( \vec{s}) = 1$ if and only if  $\exists n.\big( (s_n\in\gsem{\phi_2}_\rho) \wedge \forall m< n. (s_m\in\gsem{\phi_1}_\rho)\big)$,
\end{itemize}
\end{definition}
It is simple to verify that, for all path-formulas $\psi$, the set $\gsem{\psi}_\rho$ is Borel measurable \cite{BaierKatoenBook}. Therefore the definition is well specified. Note how the logic  $\PCTL$ can express probabilistic properties, by means of the connectives $\mathbb{P}^\forall_{\rtimes q}$ and $\mathbb{P}^\exists_{\rtimes q}$, as well as (qualitative) properties of the graph underlying the PNTS by means of the quantifiers $\forall$ and $\exists$.

\section{Encoding of  $\PCTL$}\label{section_encoding}
We prove in this section how  $\PCTL$ can be seen as a simple fragment of $\Lukmu$ by means of an explicit encoding. We first introduce a few useful macro formulas in the logic $\Lukmu$ which, crucially, are not expressible in the probabilistic $\mu$-calculus (qL$\mu$).

\begin{definition}\label{thresholds_encoding}
Let $\phi$ be a (possibly open)  $\Lukmu$ formula. We define:
\begin{center}
%\item $\boxdot\phi= (\Box\phi) \sqcap \Diamond\underline{1}$. 
$\bullet\ \  \mathbb{P}_{>0}\phi= \mu X. (X \oplus \phi) \ \ \bullet\ \ \mathbb{P}_{=1}\phi= \nu X. (X \odot \phi) \ \ \bullet\ \ \mathbb{P}_{> q}\phi=   \mathbb{P}_{>0}(\phi \odot \underline{1-q})  \ \ \bullet\ \ \mathbb{P}_{\geq q}\phi=   \mathbb{P}_{=1}(     \phi \oplus \underline{1-q}) $
\end{center}
for $q\in\mathbb{Q}\cap (0,1)$. We write $\mathbb{P}_{\rtimes q}\phi$, for $q\in\mathbb{Q}\cap [0,1]$, to denote one of the four cases.
\end{definition}
The following proposition describes the denotational semantics of these macro formulas.

\begin{proposition}\label{semantics_threshold}
Let $(S,\rightarrow)$ be a PNTS, $\phi$ a $\Lukmu$ formula and $\rho$ an interpretation of the variables. Then it holds that:
\begin{center}
$\sem{\mathbb{P}_{\rtimes q}\phi}_\rho(s)= \left\{     \begin{array}{l  l}
 						1 & $if $\sem{\phi}_\rho (s) \rtimes q \\
						0 & $otherwise$\\
						\end{array}      \right.$
\end{center}
\end{proposition}
\begin{proof}
For the case $\mathbb{P}_{>0}\phi$, observe that the map $x\mapsto q \oplus x$, for a fixed $q\!\in\![0,1]$, has $1$ as unique fixed point when $q\!>\!0$, and $0$ as the least fixed point when $q\!=\!0$. The result then follows trivially. Similarly for $\mathbb{P}_{=1}\phi$. The other cases are trivial.
\end{proof}

The following lemma is also useful.

\begin{lemma}\label{lemma_leadsto}
Let $(S,\rightarrow)$ be a PNTS, $\phi$ a $\Lukmu$ formula and $\rho$ an interpretation of the variables. Then:
\begin{itemize}
\item $\sem{\mathbb{P}_{>0}(\Diamond X)}_\rho (s) = 1$ iff\ $\ \exists t. \big(s\leadsto t \wedge \rho(X)(t)>0\big)$
\item $\sem{\mathbb{P}_{=1}(\Box X)}_\rho (s)=1$ iff\ $\  \forall t. \big(s\leadsto t \rightarrow \rho(X)(t)=1\big)$
\end{itemize}
\end{lemma}
\begin{proof}
Note that $\sem{\Diamond X}_\rho(s)>0$ iff there exists $s\rightarrow d$ such that $\displaystyle\sum_{t\in S} d(t)\rho(X)(t)>0$ holds. This is the case iff $d(t)\!>\!0$ (i.e., $s\leadsto t$) and $\rho(X)(t)\!>\!0$, for some $t\!\in\! S$. The result then follows by Proposition \ref{semantics_threshold}. The case for $\mathbb{P}_{=1}(\Box X)$ is similar.
\end{proof}

\begin{remark}\label{remark_leadsto}
When considering $\{0,1\}$-valued interpretations for $X$, the macro formula $\mathbb{P}_{>0}\Diamond$ expresses the meaning of the diamond modality in classical modal logic with respect to the graph $(S,\leadsto)$ underlying the PNTS. Similarly, $\mathbb{P}_{=1}\Box$ corresponds to the the classical box modality.
\end{remark}

We are now ready to define the encoding of  $\PCTL$ into $\Lukmu$.

\begin{definition}\label{ENCODING_PCTL_E}
We define the encoding $\mathbf{E}$ from  $\PCTL$ formulas to closed $\Lukmu$ formulas (where $\boxdot\phi$ stands for the $\Lukmu$ formula $\Box\phi \sqcap \Diamond\underline{1}$), by induction on the structure of the  $\PCTL$ formulas $\phi$ as follows:
\begin{enumerate}
\item $\label{enc_case_prop}\mathbf{E}(P)=P$,
\item $\mathbf{E}(\textnormal{true})\!=\! \underline{1}$,
\item $\mathbf{E}(\phi_1\vee \phi_2)= \mathbf{E}(\phi_1) \sqcup \mathbf{E}(\phi_2)$,
\item \label{enc_case_neg} $\mathbf{E}(\neg \phi)= \textnormal{dual}( \mathbf{E}(\phi))$,
\item \label{enc_case_exists_circ}$\mathbf{E}(\exists( \circ  \phi))= \mathbb{P}_{>0}\big(\Diamond\mathbf{E}(\phi)\big)$, 
\item \label{enc_case_forall_circ}$\mathbf{E}(\forall( \circ  \phi))= \mathbb{P}_{=1}\big(\boxdot\mathbf{E}(\phi)\big)$, 
\item \label{enc_case_exists} $\mathbf{E}(\exists(\phi_1 \ \mathcal{U} \ \phi_2))= \mu X. \Big(\mathbf{E}(\phi_2) \sqcup \big(  \mathbf{E}(\phi_1) \sqcap \mathbb{P}_{>0}(\Diamond X)\big)\Big)$,
\item \label{enc_case_forall} $\mathbf{E}(\forall(\phi_1 \ \mathcal{U} \ \phi_2) )=  \mu X. \Big(\mathbf{E}(\phi_2) \sqcup \big(  \mathbf{E}(\phi_1) \sqcap \mathbb{P}_{=1}(\boxdot X)\big)\Big)$,
\item \label{case_threshold_1} $\mathbf{E}(\mathbb{P}^{\exists}_{\rtimes q }(\circ \phi))=\mathbb{P}_{\rtimes q} \big( \Diamond\mathbf{E}(\phi)\big)$,
\item \label{case_threshold_2}$\mathbf{E}(\mathbb{P}^{\forall}_{\rtimes q }(\circ \phi))=\mathbb{P}_{\rtimes q} \big( \boxdot\mathbf{E}(\phi)\big)$,
\item \label{case_threshold_3}$\mathbf{E}(\mathbb{P}^\exists_{\rtimes q}(\phi_1 \mathcal{U} \phi_2))= \mathbb{P}_{\rtimes q} \Big( \mu X. \Big(\mathbf{E}(\phi_2) \sqcup \big(  \mathbf{E}(\phi_1) \sqcap \Diamond X\big)\Big) \Big)$,
\item \label{case_threshold_4}$\mathbf{E}(\mathbb{P}^{\forall}_{\rtimes q}(\phi_1 \mathcal{U} \phi_2))= \mathbb{P}_{\rtimes q} \Big( \mu X. \Big(\mathbf{E}(\phi_2) \sqcup \big(  \mathbf{E}(\phi_1) \sqcap \boxdot X\big)\Big) \Big)$,
\end{enumerate}
Note that Case \ref{enc_case_neg} is well defined since $ \mathbf{E}(\phi)$ is closed by construction.
\end{definition}

\begin{remark}\label{remark_fragment_plmu}
The only occurrences of \L ukasiewicz operators $\{\oplus,\odot\}$ and scalar multiplication $(q\, \phi)$ in encoded $\PCTL$ formulas appear in the formation of the macro formulas $\mathbb{P}_{\rtimes q}(\_ )$ which we refer to as \emph{threshold modalities}. Thus, $\PCTL$ can be also seen as a fragment of qL$\mu$ extended with threshold modalities as primitive operations. With the aid of these modalities the encoding is, manifestly, a straightforward adaption of the standard encoding of CTL into the  modal $\mu$-calculus (see, e.g., \cite{Stirling96}). 
\end{remark}

We are now ready to prove the correctness theorem which holds for arbitrary models.

\begin{theorem}
\label{theorem:pctl}
For every PNTS $(S,\rightarrow)$,  $\PCTL$-interpretation $\rho\!:\!\textnormal{Prop}\rightarrow(S\rightarrow\!\{0,1\})$ of the propositional letters and  $\PCTL$ formula $\phi$, the equality $\gsem{\phi}_\rho(s)=\sem{\mathbf{E}(\phi)}_\rho(s)$ holds, for all $s\in S$.
\end{theorem}
\begin{proof}[Proof (outline)]
The proof goes by induction on the complexity of $\phi$. Cases \ref{enc_case_prop}--\ref{enc_case_neg} of Definition \ref{ENCODING_PCTL_E} are trivial. Case \ref{enc_case_exists_circ} follows directly from Lemma \ref{lemma_leadsto}. Observing that $\sem{\boxdot\phi}_\rho(s)=0$ if $s\not\leadsto$ and $\sem{\boxdot\phi}_\rho(s)=\sem{\Box \phi}_\rho(s)$ otherwise, also Case \ref{enc_case_forall_circ} is a consequence of Lemma \ref{lemma_leadsto}. Consider cases \ref{enc_case_exists} and \ref{enc_case_forall}. The encoding is of the form $\mu X. (F \sqcup (G \sqcap H(X))$, where $F$ and $G$ (by induction hypothesis) and $H(X)$ (by Proposition \ref{semantics_threshold}) are all $\{0,1\}$-valued. Therefore the functor $f\mapsto \sem{F\sqcup (G\sqcap H(X))}_{\rho[f/X]}$ maps $\{0,1\}$-valued functions to $\{0,1\}$-valued functions and has only $\{0,1\}$-valued fixed-points. It then follows by Remark \ref{remark_leadsto} that the correctness of the encoding for these two cases can  be proved with the standard technique used to prove the correctness of the encoding of CTL into Kozen's $\mu$-calculus (see, e.g., \cite{Stirling96}).
Consider Case \ref{case_threshold_1}. It is immediate to verify that $\bigsqcup_{\sigma}\{ m^s_\sigma( U)\}$, where $U=\gsem{\circ \phi}_\rho=\bigcup\{ U_{\{s.t\}} \ | \ t\in \gsem{\phi}_\rho\}$, is equal (by induction hypothesis) to $\sem{\Diamond \mathbf{E}(\phi)}_\rho(s)$. The desired equality $\gsem{\mathbb{P}^{\exists}_{\rtimes q}\circ\phi}_{\rho}=\sem{\mathbb{P}_{\rtimes q }\Diamond\mathbf{E}(\phi)}_\rho$ then follows by Proposition \ref{semantics_threshold}. Case \ref{case_threshold_2} is similar. The two cases \ref{case_threshold_3} and \ref{case_threshold_4} are similar, thus we just consider case \ref{case_threshold_3}. Let $\phi= \mathbb{P}^\exists_{\rtimes q}(\psi)$ and $\psi= \phi_1 \mathcal{U} \phi_2$.  We denote with $\Psi$ the set of paths $\gsem{\psi}_\rho$. Denote by $F(X)$ the formula $\mathbf{E}(\phi_2) \sqcup (  \mathbf{E}(\phi_1) \sqcap \Diamond X)$.  It is clearly sufficient to prove that the equality $\bigsqcup_\sigma \{ m^s_\sigma(\Psi)\} = \sem{\mu X. F( X)\big) }_\rho(s)$ holds. Note that $\mu X.F(X)$ can be expressed as an equivalent qL$\mu$ formulas by substituting the closed subformulas $\mathbf{E}(\phi_1)$ and $\mathbf{E}(\phi_2)$ with two fresh atomic predicates $P_{i}$ with interpretations $\rho(P_i)=\sem{\mathbf{E}(\phi_i)}$. The equality can then be proved by simple arguments based on the game-semantics of qL$\mu$ (see, e.g., \cite{MM07} and \cite{MIO2012a}), similar to the ones used to prove that  the Kozen's $\mu$-calculus formula $\mu X. (P_2 \vee ( P_1 \wedge \Diamond X))$ has the same denotation of the CTL formula $\exists( P_1 \mathcal{U} P_2)$ (see, e.g., \cite{Stirling96}).
\end{proof}

%%%%%%%%%%%%%%%%%%%%%%%%%%%%
%%%%%%%%%%%%%%%%%%%%%%%%%%%%
%SECOND PART
%%%%%%%%%%%%%%%%%%%%%%%%%%%%
%%%%%%%%%%%%%%%%%%%%%%%%%%%%

\section{{\L}ukasiewicz $\mu$-terms}

The aim of the second half of the paper is to show how to compute the 
(rational) denotational value of a {\Lukmu} formula at any state in a finite rational probabilistic transition system. In this section, we build the main machinery for doing this, based on a system of fixed-point terms for defining 
 monotone functions from $[0,1]^n$ to $[0,1]$.
The syntax of \emph{({\L}ukasiewicz) $\mu$-terms} is specified by the grammar:
\[
t \, ::= \, 
x \mid 
% q \mid \convsum{\lambda}{t}{t} \mid 
q \, t \mid 
t \weakor t \mid t \weakand t \mid t \strongor t \mid t \strongand t \mid \Mu{x}{t} \mid \Nu{x}{t}
\]
Again, $q$ ranges over rationals in $[0,1]$. As expected, the $\mu$ and $\nu$ operators bind their variables. We write $t(x_1, \dots, x_n)$ to mean that all free variables of $t$ are contained in $\{x_1, \dots, x_n\}$. 

The \emph{value} $t(\vec{r})$ (we eschew semantic brackets) of a $\mu$-term $t(x_1, \dots, x_n)$ applied to a vector 
$(r_1, \dots, r_n) \in [0,1]^n$ is defined inductively in the obvious way, cf.\ Definition~\ref{definition:semantics}. (Indeed,  $\mu$-terms form a fragment of $\Lukmu$ of formulas whose value is independent of the transition system in which they are interpreted.)

%\begin{align*}
%x_i(\vec{r}) & = r_i
%\\
%(q\, t)(\vec{r}) &  = q \times t({\vec{r}})
%\\
%(t_1 \weakor t_2)(\vec{r}) & = \max(t_1(\vec{r}), t_2(\vec{r}))
%\\
%(t_1 \weakand t_2)(\vec{r}) & = \min(t_1(\vec{r}), t_2(\vec{r}))
%\\
%(t_1 \strongor t_2)(\vec{r}) & = \min(t_1(\vec{r})+ t_2(\vec{r}),\, 1)
%\\
%(t_1 \strongand t_2)(\vec{r}) & = \max(t_1(\vec{r})+ t_2(\vec{r}) - 1, \, 0)
%\\
%(\Mu{x_{n+1}}{t})(\vec{r}) & = \lfp (r_{n+1} \mapsto  t(\vec{r},r_{n+1}))
%\\
%(\Nu{x_{n+1}}{t})(\vec{r}) & = \gfp (r_{n+1} \mapsto  t(\vec{r},r_{n+1}))
%\end{align*}
%\noindent
%The fixed-points, in the clauses for $\mu$ and $\nu$, exist because $t(x_1, \dots, x_n)$ always defines a monotone function from $[0,1]^n$ to $[0,1]$ (by routine induction on $t$).

In Section~\ref{section:model-checking}, the model-checking task will be 
reduced to the problem of computing the value of $\mu$-terms. 
% In this section and  Section~\ref{section:algorithm}, we show how to compute such values. 
The fundamental property that allows such values to be computed is that, for any $\mu$-term $t(x_1, \dots, x_n)$ and vector of rationals $(q_1, \dots, q_n)$, the value of $t(\vec{q})$ is rational and can be computed from $t$ and $q$. One way of establishing this result is by a simple reduction to the first-order theory of rational linear arithmetic, which provides an indirect means of computing the value of $t(\vec{q})$. The current section presents a brief outline of this approach.
After this, in Section~\ref{section:algorithm}, we provide an alternative
direct algorithm for computing $t(\vec{q})$.

A \emph{linear expression} in variables $x_1, \dots, x_n$ is
an expression 
\[
q_1 x_1 + \dots + q_n x_n + q
\]
where $q_1, \dots, q_n, q$ are real numbers. In the sequel, we only consider \emph{rational} linear expressions, in which $q_1, \dots, q_n, q$ are all rational,
%that are \emph{rational} (i.e., $q_1, \dots, q_n, r$ are rational), 
% and \emph{monotone} (i.e., $q_1, \dots q_n \geq 0$), 
and we henceforth assume 
% these two properties without further mentioning them explicitly.
this property without mention.
We write $e(x_1, \dots, x_n)$ if $e$ is a linear expression in $x_1, \dots, x_n$, in which case, given real numbers $r_1, \dots, r_n$, we write $e(\vec{r})$ for the value of the expression when the variables $\vec{x}$ take values $\vec{r}$.
We also make use of the closure of linear expressions under substitution: given $e(x_1, \dots, x_n)$ and $e_1(y_1, \dots, y_m), \dots, e_n (y_1, \dots, y_m)$,
we write $e(e_1, \dots, e_n)$ for the evident substituted expression in variables $y_1,\dots, y_m$ (which is defined formally by multiplying out and adding coefficients).

The first-order 
theory of \emph{rational linear arithmetic} has 
linear expressions as terms, and strict and non-strict 
inequalities between linear expressions,
\begin{equation}
\label{eqn:inequalities}
e_1 < e_1 \qquad e_1 \leq e_2 \enspace ,
\end{equation}
as atomic formulas. 
Equality can be expressed as the conjunction of two non-strict inequalities and the negation of an atomic formula can itself be expressed as an atomic formula. 
The truth of a first-order formula is given via its interpretation in the reals, or equivalently in the rationals since the inclusion of the latter in the former is an elementary embedding. The theory enjoys quantifier elimination~\cite{Ferrante1975}.

\begin{proposition}
\label{proposition:mu-term:arithmetic}
For every {\L}ukasiewicz $\mu$-term $t(x_1, \dots, x_n)$, 
its graph  
$\{(\vec{x},y) \in [0,1]^{n+1} \mid t(\vec{x}) = y\}$ is definable by a formula $F_t(x_1, \dots, x_n, y)$ in the first-order theory of rational linear arithmetic, where $F_t$ is computable from $t$.
\end{proposition}

\begin{proof}
The proof is a straightforward induction on the structure of $t$. We consider two cases, in order to illustrate the simple  manipulations  used in the construction of $F_t$.

If $t$ is $t_1 \strongor t_2$ then $F_t$ is the formula
\[
\exists z_1, z_2 .\: F_{t_1}(\vec{x},z_1) \, \wedge \,
  F_{t_2}(\vec{x},z_2) \, \wedge \,
  \left((z_1 + z_2 \leq 1 \wedge z = z_1 + z_2) \, \vee \, 
    (1\leq z_1 + z_2  \wedge z = 1)
  \right)
\]

If $t$ is $\Mu{x_{n+1}}{t'}$ then $F_t$ is the formula
\[
F_{t'}(x_1, \dots, x_n, y,y) \, \wedge \, \forall z.\: F_{t'}(x_1, \dots, x_n, z,z) \Implies y \leq z \enspace .
\]
\end{proof}
\noindent
Proposition~\ref{proposition:mu-term:arithmetic} provides the following method of computing the value $t(\vec{q})$ of  $\mu$-term $t(x_1, \dots, x_n)$ at a rational vector $(q_1, \dots, q_n) \in [0,1]^n$. First construct $F_t(x_1, \dots, x_n,y)$. Next, perform quantifier elimination to obtain an equivalent 
quantifier-free formula $G_t(x_1, \dots, x_n,y)$, and consider its 
instantiation $G_t(q_1, \dots, q_n,y)$ at $\vec{q}$. (Alternatively, obtain 
an equivalent formula $G^{\vec{q}}_t(y)$ by performing quantifier elimination
on $F_t(q_1, \dots, q_n,y)$.) By performing obvious simplifications of atomic formulas in one variable, $G_t(q_1, \dots, q_n,y)$ reduces to a boolean combination of inequalities each having one of the following forms
\[
y \leq q \qquad y < q \qquad y \geq q \qquad y > q \enspace .
\]
By the correctness of $G_t$ there must be a unique rational satisfying the boolean combination of constraints, and this can be extracted in a straightforward way from $G_t(q_1, \dots, q_n,y)$. 

We give a crude (but sufficient for our purposes) complexity analysis of the above procedure. 
%Clearly, if
%$l_1,l_2$ are the lengths of the formulas $F_{t_1}, F_{t_2}$ respectively, then the length of, for example, $F_{t_1 \strongor t_2}$  is
%\[
%l \; =  \; l_1 + l_2 + c \enspace ,
%\]
%and if $l'$ is the length of $F_{t'}$, for $t(x_1, \dots, x_{n+1})$, then the 
%length of $F_{\Mu{x_{n+1}}{t'}}$ is
%\[
%l \: \leq \: 2l' + c \enspace .
%\]
In general, for a $\mu$-term $t$ of length $u$ containing $v$ fixed points, the
length of $F_t$ is bounded by
$2^v u c$, for some constant $c$.
The quantifier-elimination procedure in~\cite{Ferrante1975}, when given a formula of length $l$ as input produces a formula of length at most $2^{dl}$ as output, for some constant $d$, and takes time at most $2^{2^{d'l}}$.
Thus the length of the formula $G_t(x_1, \dots, x_n,y)$ is
bounded by $2^{2^v u c d}$, 
and the computation time for $t(\vec{q})$  is $\OO{2^{2^{2^v u c d'}}}$, using a unit cost model for rational arithmetic. 

%The above bounds could be improved using a more refined analysis of quantifier elimination (for example, taking quantifier alternation into account), coupled with a more careful translation from $\mu$-terms to linear arithmetic (for example, solving adjacent fixed points of the same kind simultaneously to avoid redundant  quantifier elimination). Nevertheless, even with such improvements, the end formula will still involve doubly-exponential space and triply-exponential time. Furthermore, the crude analysis above is sufficient for comparison with the direct algorithm described in the next section.

\section{A direct algorithm for evaluating $\mu$-terms}
\label{section:algorithm}

Our direct approach to computing the values of $\mu$-terms is based on a
simple explicit representation of the functions  defined by such terms.
A \emph{conditioned linear expression} is  a pair, written 
$\GLE{C}{e}$,
where $e$ is a linear expression, and $C$ is a finite set of strict and non-strict inequalities between linear expressions; i.e., each element of
$C$ has one of the forms in~(\ref{eqn:inequalities}).
We write $C(\vec{r})$ for the conjunction of the inequations obtained by instantiating $\vec{r}$ for 
$\vec{x}$ in $C$. Clearly, if $\vec{q}$ is a vector of rationals then it is decidable if $C(\vec{q})$ is true or false. 
The intended meaning of a conditioned linear expression $\GLE{C}{e}$ is that it 
denotes the value $e(\vec{r})$ when applied to a vector of reals $\vec{r}$  for which $C(\vec{r})$ is true, otherwise it is undefined.
A basic property we exploit in the sequel is that every conditioning set $C(x_1, \dots, x_n)$ defines a convex subset $\{(r_1, \dots, r_n) \mid C(\vec{r})\}$ of $\mathbb{R}^n$.

%Of particular interest to us are systems of conditioned linear expresssions in variables $x_1, \dots, x_n$ that represent functions  from $[0,1]^n$ to $[0,1]$. As an example, in one variable, consider:
%\begin{align*}
%\GLE{x < \frac{1}{2}\: &}{\:\frac{1}{2}\, x + \frac{1}{8}}
%\\
%\GLE{x = \frac{1}{2}\: &}{\:\frac{7}{16}}
%\\
%\GLE{x > \frac{1}{2}\: &}{\:x - \frac{1}{32}}
%\end{align*}
%(the conditioning equality is a shorthand for the conjunction of two inequalities). This example shows that the representable functions need not be continuous.
%Note that the represented function is monotone and has unique fixed point, $\frac{1}{4}$.

%In general, 
Let $\mathcal{F}$ be a \emph{system} (i.e., finite set) of conditioned linear expresssions in variables $x_1, \dots, x_n$. We say that $\mathcal{F}$ \emph{represents} a function $f \colon [0,1]^n \to [0,1]$ if the following conditions hold:
\begin{enumerate}
\item For all $d_1, \dots, d_n \in [0,1]$, there exists a conditioned linear expression $(\GLE{C}{e}) \in \mathcal{F}$ such that $C(\vec{d})$ is true, and

\item for all $d_1, \dots, d_n \in [0,1]$, and every conditioned linear expression $(\GLE{C}{e}) \in \mathcal{F}$, if  $C(\vec{d})$ is true then $e(\vec{d}) = f(\vec{d})$.

\end{enumerate}

\noindent
Note that, for two conditioned linear expressions $(\GLE{C_1}{e_1}), (\GLE{C_2}{e_2}) \in \mathcal{F}$, we do not require different conditioning sets $C_1$ and $C_2$ to be disjoint. However, $e_1$ and $e_2$ must agree on any overlap.

Obviously, the function represented by a system of conditioned linear expressions is unique, when it exists. But not every system represents a function.
One could impose syntactic conditions on a system to ensure that it represents a function, but we shall not pursue this.

While conditioned linear expressions provide a syntax more directly tailored to expressing functions than general logical formulas, their expressivity in this regard coincides with rational linear arithmetic.

\begin{proposition}
\label{proposition:cle:arithmetic}
A function $f\colon [0,1]^n \to [0,1]$ is representable by a system of conditioned linear expressions if and only if its graph 
$\{(\vec{x},y) \in [0,1]^{n+1} \mid f(\vec{x}) = y\}$ is definable by a formula $F(x_1, \dots, x_n, y)$ in the first-order theory of rational linear arithmetic. Moreover, a defining formula and a representing system of conditioned linear equations can each be computed from the other.
\end{proposition}
\noindent
We believe this result to be folklore. The proof is a straightforward application of quantifier elimination.

Combining Propositions~\ref{proposition:mu-term:arithmetic} and~\ref{proposition:cle:arithmetic} we obtain:

\begin{corollary}
\label{corollary:sledgehammer}
For every {\L}ukasiewicz $\mu$-term $t(x_1, \dots, x_n)$, the function 
\[
\vec{r} \mapsto t(\vec{r}) \colon [0,1]^n \to [0,1]
\]
is representable by a system of conditioned linear expressions in variables $x_1, \dots, x_n$. Furthermore a representing system can be computed from $t$.
\end{corollary}

The computation of a representing system for $t$ via quantifier elimination,
provided by the proofs of
Propositions~\ref{proposition:mu-term:arithmetic} and~\ref{proposition:cle:arithmetic}, is indirect.
The goal of this section is to present an alternative algorithm for calculating
the value $t(\vec{r})$ of a $\mu$-term at rationals $r_1, \dots, r_n  \in [0,1]$, which is
directly based on manipulating conditioned linear expressions. 
Rather than computing an entire system of conditioned linear expressions representing $t$, the algorithm works locally to provide a single conditioned expression that applies to the input vector $\vec{r}$.

The algorithm takes, as input, a $\mu$-term $t(x_1, \dots, x_n)$ and a vector of rationals $(r_1, \dots, r_n) \in [0,1]^n$, and returns a conditioned linear expression $\GLE{C}{e}$, in variables $x_1, \dots, x_n$, with the following two properties.
\begin{itemize}
\item[(P1)] $C(\vec{r})$ is true.

\item[(P2)] For all $s_1, \dots, s_n \in \mathbb{R}$, if $C(\vec{s})$ is true then
$s_1, \dots, s_n \in [0,1]$ and $e(\vec{s}) =  t(\vec{s})$.
\end{itemize}
\noindent
It follows that $e(\vec{r}) =  t(\vec{r})$, so $e$ can indeed be used to compute the value $t(\vec{r})$.

%The algorithm is presented in detail in the next subsection, with justification of the individual steps interposed. 
%It is followed by a further subsection in which the termination property needed to complete the proof of correctness is established.

\subsection{The algorithm}

The algorithm takes, as input, a $\mu$-term $t(x_1, \dots, x_n)$ and a vector of rationals $(r_1, \dots, r_n) \in [0,1]^n$, and returns a conditioned linear 
expression $\GLE{C}{e}$, in variables $x_1, \dots, x_n$, 
with the properties (P1) and (P2) above. For the purposes of the correctness proof in Section~\ref{section:correctness}, it is convenient to consider the running of the algorithm in the more general case that $r_1, \dots, r_n$ are arbitrary real numbers in $[0,1]$. This more general algorithm can be understood as an algorithm in the Real RAM (a.k.a.\ BSS) model of computation~\cite{BSS1989}. When the input vector is rational, all real numbers encountered during execution of the algorithm are themselves rational, and so the general Real RAM algorithm specialises to  a 
\emph{bona fide} (Turing Machine) algorithm in this case. Moreover, even in the case of irrational inputs, all linear expressions constructed in the course of the algorithm are rational.

The algorithm works recursively on the structure of the term $t$. We present illustrative cases for terms
$t_1 \strongor t_2$ and $\Mu{x_{n+1}}{t'}$. The latter is the critical case. 
The algorithm for $\Nu{x_{n+1}}{t'}$ is an obvious dualization.

If $t$ is $t_1 \strongor t_2$ then recursively compute $\GLE{C_1\,}{\,e_1}$ and
 $\GLE{C_2\,}{\,e_2}$. If $e_1(\vec{r}) + e_2(\vec{r}) \leq 1$ then return
\[
\GLE{C_1, \, C_2, \, e_1 + e_2 \leq 1  \, }{\, e_1 + e_2} \enspace .
\]
Otherwise, return
\[
\GLE{C_1, \, C_2, \, e_1 + e_2 \geq 1  \, }{\, 1} \enspace .
\]

In the case that $t$ is $\Mu{x_{n+1}}{t'}$, enter the following 
loop starting with $D = \emptyset$ and $d =0$. 

\paragraph{Loop:}
At the entry of the loop we have a finite set $D$ of inequalities between linear expressions in $x_1, \dots, x_n$, and we have a linear expression $d(x_1, \dots, x_n)$.
The loop invariant that applies is:
\begin{itemize}
\item[(I1)] $D(\vec{r})$ is true; and
\item[(I2)]  for all $\vec{s} \in [0,1]^n$, if $D(\vec{s})$ then $d(\vec{s}) \leq (\Mu{x_{n+1}}{t'})(\vec{s})$.
\end{itemize}
 \noindent
We think of $D$ as constraints propagated from earlier iterations of the loop, and of $d$ as the current approximation to the least fixed point subject to the constraints.

Recursively compute $t'(x_1, \dots, x_{n+1})$ at $(\vec{r},d(\vec{r}))$ as $\GLE{C}{e}$, where $e$ has the form:
\begin{equation}
\label{equation:e-form}
q_1\, x_1 + \dots + q_n \, x_n + q_{n+1}\,  x_{n+1} + q \enspace .
\end{equation}

In the case that $q_{n+1} \neq 1$, 
define the linear expression:
\begin{equation}
\label{equation:f}
f \: := \: \frac{1}{1-q_{n+1}} \, \left( \, q_1 \, x_1 + \dots + q_n \, x_n  + q \, \right) \enspace .
\end{equation}
Test if $C(\vec{r},f(\vec{r}))$ is true. If it is, 
exit the loop and return:
\begin{equation}
\label{equation:first-result}
\GLE{D \, \cup \, C(x_1, \dots, x_n, d(x_1, \dots, x_n) )\, \cup \, C(x_1, \dots, x_n, f(x_1, \dots, x_n) )\,}{\, f} % \enspace .
\end{equation}
as the result of the algorithm for $\Mu{x}{t'}$ at $\vec{r}$.
Otherwise, if 
$C(\vec{r},f(\vec{r}))$ is  false, define  $N(x_1, \dots, x_n)$ to be the negation of the inequality $e_1(x_1, \dots, x_n, f(x_1, \dots x_n)) \vartriangleleft e_2(x_1, \dots, x_n, f(x_1, \dots x_n))$ (using 
$\vartriangleleft$ to stand for either $<$ or $\leq$), where
$e_1(x_1, \dots, x_{n+1})  \vartriangleleft e_2(x_1, \dots, x_{n+1})$ is a chosen inequality in $C$
for which $e_1(\vec{r}, f(\vec{r})) \vartriangleleft e_2(\vec{r}, f(\vec{r}))$ is false, and
go to \textbf{find next approximation} below.

In the case that $q_{n+1} = 1$, test the 
equality
%\begin{equation}
%\label{test}
$q_1 \, r_1 + \dots + q_n \, r_n  + q  = 0$.
%\end{equation}
If true, exit the loop with result:
\begin{equation}
\label{equation:second-result}
\GLE{D  \, \cup \, C(x_1, \dots, x_n, d(x_1, \dots, x_n) ) \, \cup \, \{q_1 \, x_1 + \dots + q_n \, x_n  + q  =  0\} \,} {\, d} \enspace .
\end{equation}
If instead 
$q_1 \, r_1 + \dots + q_n \, r_n  + q  \neq 0$, choose $N(x_1, \dots, x_n)$ to be whichever of 
the inequalities
\[
q_1 \, x_1 + \dots + q_n \, x_n  + q  \; <  \; 0 \qquad
0 \; < \; q_1 \, x_1 + \dots + q_n \, x_n  + q  
\]
is true for $\vec{r}$, 
and 
proceed with
\textbf{find next approximation} below.

\paragraph{Find next approximation:}
Arrange the inequalities in $C$ so they have the following structure. 
\begin{equation}
\label{equation:C}
C' 
\, \cup \, 
\{ x_{n+1} > a_i \}_{1 \leq i \leq l'} 
\, \cup \, 
\{ x_{n+1} \geq a_i \}_{l' < i \leq l} 
\, \cup \, 
\{ x_{n+1} \leq b_i \}_{1 \leq i \leq m'} 
\, \cup \, 
\{ x_{n+1} < b_i \}_{m' < i \leq m} 
\end{equation}
such that the only variables in the inequalities $C'$, and linear expressions $a_i, b_i$ are $x_1, \dots, x_n$. 
Choose $j$ with $1 \leq j \leq m$ such that 
$b_j(\vec{r}) \leq b_i(\vec{r})$ for all $i$ with $1 \leq i \leq m$. 
Then go back to \textbf{loop}, %for the next iteration, 
taking
\begin{equation}
\label{equation:next-iteration}
D \, \cup \, C(x_1, \dots, x_n, d(x_1, \dots, x_n))\,  \cup \, \{N(x_1, \dots, x_n)\}
\, \cup \, \{b_j \leq b_i \mid 1 \leq i \leq m\}
\qquad
e(\vec{x},b_j(\vec{x}))
\end{equation}
to replace $D$ and $d$ respectively.

%\subsection{Simple Example 1}

%Consider the $\Lukmu$ term $t=\mu x. (   \frac{1}{2}x \oplus \frac{1}{2})$. Let $t^\prime(x)=  \frac{1}{2}x \oplus \frac{1}{2}$ so that $t=\mu x.t^\prime (x)$.
%A naive iteration for computing $t$ would require $\omega$-steps: $0$, $t^\prime(0)=\frac{1}{2}$, $t^\prime(\frac{1}{2})=\frac{3}{4}$, $\dots$. More generally, computing fixed points values by  iterations requires more than $\omega$ steps since $\Lukmu$ terms are not necessarily continuous functions.

%The algorithm applied on $t^\prime$ produces $C\vdash e$ as $\{ 0\leq x, x\leq 1\}\vdash  \frac{1}{2}x+ \frac{1}{2}$, on all inputs $r\in\mathbb{R}$.
%Start the algorithm on $t$ (empty vector $\emptyset$ as input) with $D=\emptyset$ and $d=0$. The term $e$ is of the form $q_1 x + q$ with $q_x=q=\frac{1}{2}$. Thus we are in the case $q_n+1 \neq 0$ of Equation (\ref{equation:crucial}), for $n=0$. Calculate $f:= \frac{1}{1-\frac{1}{2}}(\frac{1}{2})=1$. We need to verify if $C(1)$ holds. It does, thus $1$ is the desired value of $t$.

\subsection{A simple example}
\label{section:example}
Consider the $\Lukmu$ term $t=\mu x. (     \mathbb{P}_{\geq \frac{1}{2}}x \, \sqcup \, \frac{1}{2})$, where $\mathbb{P}_{\geq \frac{1}{2}}x$ is the macro formula as in Definition~\ref{thresholds_encoding}, that is   $\mathbb{P}_{\geq \frac{1}{2}}x= \mathbb{P}_{=1}(    x \oplus \frac{1}{2}) = \nu y. ( y \odot (x \oplus  \frac{1}{2}))$. Thus,

$$
t=\mu x. \Big(  \nu y. \big( y \odot (x \oplus  \frac{1}{2})\big)  \sqcup \frac{1}{2}\Big)
$$
Here, $t^\prime(x)=  \nu y. \big( y \odot (x \oplus  \frac{1}{2})\big)  \sqcup \frac{1}{2}$ is a discontinuous function, and the value of $t$ is $1$.

We omit giving a detailed simulation of the algorithm on the subexpression $t'(x)$ at $x=r$. The result it produces, however,  is
$\{Ê0\leq x < \frac{1}{2}\}\vdash \frac{1}{2}$ if ${r}<\frac{1}{2}$, and 
$\{Ê\frac{1}{2} \leq x \leq 1\}\vdash 1$ if $r\geq \frac{1}{2}$.

We run the algorithm on input $\mu x.t^\prime(x)$. Set $D=\emptyset$ and $d=0$. Calculating $t^\prime(x)$ at $x = 0$ we obtain $C\vdash e$ as $\{Ê0\leq x < \frac{1}{2}\}\vdash \frac{1}{2}$.
We now need to calculate $f:=\frac{1}{1-0}(\frac{1}{2})=\frac{1}{2}$. The constraint $C(\frac{1}{2})$ does not hold. Thus we need to improve the approximation $d=0$.
Since $e=\frac{1}{2}$ is constant, the next approximation is $\frac{1}{2}$. The new set of constraints is still the emptyset. Thus we iterate the algorithm with $D=\emptyset$ and $d=\frac{1}{2}$. Calculating $t^\prime(x)$ at $x=\frac{1}{2}$ produces $C\vdash e$ as $\{Ê\frac{1}{2} \leq  x \leq 1\}\vdash 1$. Compute $f:=\frac{1}{1- 0}(1)=1$. Since $C(1)$ holds, the algorithm terminates with $\GLE{\emptyset}{1}$, as desired.

\subsection{Correctness of the algorithm}
\label{section:correctness}

%The correctness of the individual steps of the algorithm has been justified above. It remains only to prove termination.
%
\begin{theorem}
\label{theorem:algorithm-correct}
Let $t(x_1, \dots, x_n)$ be any {\L}ukasiewicz $\mu$-term.
Then, for every input vector $(r_1, \dots, r_n) \in [0,1]^n$, 
the above (Real RAM) algorithm terminates with a conditioned linear expression
$\GLE{C_{\vec{r}} \,}{\, e_{\vec{r}}}$ satisfying properties (P1) and (P2). 
Moreover, the set of all 
possible resulting conditioned linear expressions
\begin{equation}
\label{equation:result-set}
\{ \GLE{C_{\vec{r}} \,}{\, e_{\vec{r}}} \mid \vec{r} \in [0,1]^n \}
\end{equation}
is finite, and thus provides a representing system for the function $t \colon [0,1]^n \to [0,1]$.
\end{theorem}
\noindent
Before the proof it is convenient to introduce some terminology associated with the properties stated in the theorem.
For a $\mu$-term $t$, we call the cardinality of the set
(\ref{equation:result-set}) of possible results, $\GLE{C_{\vec{r}} \,}{\, e_{\vec{r}}}$,
the  \emph{basis size},
and we call the maximum number of inequalities in any $C_{\vec{r}}$ the 
\emph{condition size}.

\begin{proof} %[Proof of Theorem~\ref{theorem:algorithm-correct}]
By induction on the structure of $t$. We verify the critical case when $t$ is $\Mu{x_{n+1}}{t'}$.

We show first that the loop invariants (I1), (I2) guarantee that any result  returned via (\ref{equation:first-result}) or (\ref{equation:second-result}) satisfies (P1) and (P2).
By induction hypothesis, the recursive computation of 
$t'(x_1, \dots, x_{n+1})$ at $(\vec{r},d(\vec{r}))$ as $\GLE{C}{e}$, where 
$e$ has the form 
$q_1\, x_1 + \dots + q_n \, x_n + q_{n+1}\,  x_{n+1} + q$  as in (\ref{equation:e-form}), satisfies:
$C(\vec{r}, d(\vec{r}))$; and, for all
$s_1, \dots, s_{n+1} \in \mathbb{R}$, if 
$C(s_1, \dots, s_{n+1})$ then $\vec{s} \in [0,1]^n$ and 
$t'(s_1, \dots, s_{n+1}) = e(s_1, \dots, s_{n+1})$. 

In the case that $q_{n+1} \neq 1$, the linear expression $f$, defined in
(\ref{equation:f}), maps any $s_1, \dots, s_n \in \mathbb{R}$ to the unique solution
$f(\vec{s})$ to the equation $x_{n+1} = e(s_1, \dots, s_n, x_{n+1})$ in $\mathbb{R}$.
Suppose that $D(\vec{s})$ holds.
Then, by  loop invariant (I2), $d(\vec{s}) \leq (\Mu{x_{n+1}}{t'})(\vec{s})$.
Suppose also that $C(\vec{s}, f(\vec{s}))$. Then 
$t'(\vec{s},f(\vec{s})) = e(\vec{s},f(\vec{s})) = f(\vec{s})$, i.e., 
$f(\vec{s})$ is a fixed point of $x_{n+1} \mapsto t'(\vec{s}, x_{n+1})$; whence, 
$(\Mu{x_{n+1}}{t'})(\vec{s}) \leq f(\vec{s})$.
Suppose, finally, that $C(\vec{s}, d(\vec{s}))$ also holds. Then, because
both $C(\vec{s}, d(\vec{s}))$ and $C(\vec{s}, f(\vec{s}))$, and 
$d(\vec{s}) \leq (\Mu{x_{n+1}}{t'})(\vec{s}) \leq f(\vec{s})$, we have, by the convexity of constraints, that %$C(\vec{s}, (\Mu{x_{n+1}}{t'})(\vec{s}))$.
$t'(\vec{s}, s_{n+1}) = e(\vec{s}, s_{n+1})$ for all $s_{n+1} \in [d(\vec{s}),f(\vec{s})]$. So $f(\vec{s})$ is the unique fixed-point of
$x_{n+1} \mapsto t'(\vec{s}, x_{n+1})$ on $[d(\vec{s}),f(\vec{s})]$. Since, $d(\vec{s}) \leq (\Mu{x_{n+1}}{t'})(\vec{s})$, 
% there is no fixed-point of
% $x_{n+1} \mapsto t'(\vec{s}, x_{n+1})$ below $d(\vec{s})$. So 
we have $f(\vec{s}) = (\Mu{x_{n+1}}{t'})(\vec{s})$. This argument justifies that the conditioned linear expression of 
(\ref{equation:first-result}) satisfies (P2). It satisfies
(P1) just if $C(\vec{r},f(\vec{r}))$, which is exactly the condition under which (\ref{equation:first-result}) is returned as the result.

In the case that $q_{n+1} = 1$ then, for any  $s_1, \dots, s_n \in \mathbb{R}$, 
the equation $x_{n+1} = e(s_1, \dots, s_n, x_{n+1})$  has a solution if and only if
$q_1 \, s_1 + \dots + q_n \, s_n  + q  = 0$, 
in which case any $x_{n+1} \in \mathbb{R}$ is a solution. 
Suppose that $q_1 \, s_1 + \dots + q_n \, s_n  + q  = 0$ and 
$C(\vec{s}, d(\vec{s}))$ both hold. Then $t'(s_1, \dots, s_n,d(\vec{s})) =
e(\vec{s}, d(\vec{s})) = d(\vec{s})$, so $d(\vec{x})$ is a fixed point of 
$x_{n+1} \mapsto t'(\vec{s}, x_{n+1})$. If also $D(\vec{s})$ holds then, by
 loop invariant (I2), $d(\vec{x}) = (\Mu{x_{n+1}}{t'})(\vec{s})$.
We have justified that the conditioned linear expression of 
(\ref{equation:second-result}) satisfies (P2). It satisfies
(P1) just if $q_1 \, r_1 + \dots + q_n \, r_n  + q  = 0$,
which is exactly the condition under which (\ref{equation:second-result}) is returned as the result.

Next we show that the loop invariants are preserved through the computation. 
Properties (I1) and (I2) are trivially satisfied by the initial values 
$D = \emptyset$ and $d =0$. We must show that they are preserved when $D$ and $d$ are modified via (\ref{equation:next-iteration}), which happens when 
execution passes to \textbf{find next approximation}.
In this subroutine, the inequalities  in $C$ are first arranged as in (\ref{equation:C}) where,
as $C(\vec{r},d(\vec{r}))$, we must have  $m \geq 1$, as otherwise $C(\vec{r},s)$ would hold for all real $s \geq  d(\vec{r})$, contradicting that $C(\vec{r},s)$ implies $s \in [0,1]$. (Similarly, $l \geq 1$.)
%Suppose, for contradiction, that $m = 0$.
%Then, since $C(\vec{r},d(\vec{r}))$, and all inequalities in $C$ bound $x_{n+1}$ from below, we have that $C(\vec{r},y)$ for all $y \in [d(\vec{r}), 1]$. But then, since 
%the least fixed point $r'$ of $x_{n+1} \mapsto {t'}(\vec{r},x_{n+1})$ satisfies $r' \geq d(\vec{r})$, it holds that $C(\vec{r},r')$, contradicting the property we observed at the start of the \textbf{find next approximation} subroutine.
%
Thus there indeed exists $j$ with $1 \leq j \leq m$ such that 
$b_j(\vec{r}) \leq b_i(\vec{r})$ for all $i$ with $1 \leq i \leq m$. 
It is immediate that the constraints in the modified
$D$  of (\ref{equation:next-iteration}) 
are true for $\vec{r}$. Thus (I1) is preserved.
To show (I2), 
suppose $s_1, \dots, s_n$ satisfy the constraints, i.e.,
\begin{equation*}
%\label{equation:for-argument}
D(\vec{s})  \qquad  C(\vec{s}, d(\vec{s})) \qquad
N(\vec{s}) \qquad 
\{b_j(\vec{s}) \leq b_i(\vec{s}) \mid 1 \leq i \leq m\} \enspace .
\end{equation*}
Defining $r' = (\Mu{x_{n+1}}{t'})(\vec{s})$, by (I2) for $D,d$ we have 
$d(\vec{s}) \leq r'$. We must show that $e(\vec{s},b_j(\vec{s})) \leq r'$.
By the definition of  $N(x_1, \dots, x_{n})$, in either the 
$q_{n+1} \neq  1$ or  $q_{n+1} = 1$ case, $N(\vec{s})$ implies that 
$C(\vec{s},\, r')$ does not hold.
Because $C(\vec{s},d(\vec{s}))$ and by the choice of $j$, it holds that 
$C(\vec{s},s)$, for all $s \in [0,1]$ such that  $s = d(\vec{s})$  or $d(\vec{s}) < s < b_j(\vec{s})$.
%(The range of $s_{n+1}$ is expressed as a disjunction in order to include 
%$d(\vec{s})$ in the case that $d(\vec{s}) =  b_j(\vec{s})$. The motivation for using such approximations $s_{n+1}$, in the first place, is to have 
%a uniform argument that applies irrespective of whether $C(\vec{s}, b_j(\vec{s}))$ holds or not, avoiding a case analysis on whether the relevant  upper-bounding inequalities in $C$ are strict or non-strict.)
Since $C(\vec{s},r')$ is false and $d(\vec{s}) \leq r'$, it follows from the convexity of the conditioning set $C$ that, for every $s$
with  $s = d(\vec{s})$  or $d(\vec{s}) < s < b_j(\vec{s})$, we have $s < r'$.  
Whence, since $r'$ is the least prefixed point for $x_{n+1} \mapsto {t'}(\vec{s},x_{n+1})$, 
also $s < {t'}(\vec{s},s) \leq r'$, i.e., 
\begin{equation}
\label{equation:little}
s < e(\vec{s},s) \leq r' \enspace .
\end{equation}
Thus, $e(\vec{s},b_j(\vec{s})) = \sup \{ e(\vec{s},s) \mid 
\text{$s = d(\vec{s})$ or $d(\vec{s}) \leq s < b_j(\vec{s})$} \} \leq r'$.
Thus, $e(\vec{s},b_j(\vec{s})) \leq r'$, i.e., it is an approximation to the fixed point.
Moreover, it is a good new approximation to choose in the sense that:
\begin{equation}
\label{equation:better}
d(\vec{s}) < e(\vec{s},b_j(\vec{s}))~\text{~and ~ not}~
C(\vec{s}, e(\vec{s},b_j(\vec{s}))) \enspace .
\end{equation}
The former holds because $d(\vec{s}) < e(\vec{s},d(\vec{s}))$, by
(\ref{equation:little}),  and $d(\vec{s}) \leq b_j(\vec{s})$. The latter because if
$C(\vec{s}, e(\vec{s},b_j(\vec{s})))$ then, in particular, 
$e(\vec{s},b_j(\vec{s})) \leq b_j(\vec{s})$, so
$b_j(\vec{s}) = e(\vec{s},b_j(\vec{s})) = r'$, contradicting that 
not $C(\vec{s},r')$.

To show termination, by induction hypothesis, collecting all possible results of running the 
algorithm on $t'$ produces a representing system for $t' \colon [0,1]^{n+1} \to [0,1]$:
\begin{equation}
\label{equation:list}
\GLE{C_1\,}{\,e_1} \quad \dots \quad \GLE{C_{k'}\,}{\,e_{k'}} \enspace ,
\end{equation}
where $k'$ is the basis size of $t'$.
We now analyse the execution of the algorithm for $\Mu{x_{n+1}}{t'}$  on a given input vector $(r_1, \dots, r_n)$. On iteration number $i$, the loop is 
entered with constraints $D_i$ and approximation $d_i$ (where 
$D_1 = \emptyset$ and $d_1 = 0$), after which the recursive call to the 
algorithm for $t'$ yields one
of the conditioned linear expressions, $\GLE{C_{k_i}\,}{\,e_{k_i}}$, from
(\ref{equation:list}) above,
such that
$C_{k_i}(\vec{r}, d_i(\vec{r}))$ holds. Then, depending on conditions involving 
only $\GLE{C_{k_i}\,}{\,e_{k_i}}$ and $\vec{r}$, either a 
result is returned, or $D_{i+1}$ and $d_{i+1}$ are constructed for the 
loop to be repeated.
By (\ref{equation:better}), at iteration $i+1$ of the loop, we have
$d_{i+1} (\vec{r}) > d_i(\vec{r})$ and also $C_{k_i}(\vec{r}, d_{i+1}(\vec{r}))$ is false. 
Since each conditioning set
%of the sets $\{(s_1, \dots, s_{n+1}) \mid C_j(s_1, \dots, s_{n+1})\}$ 
is convex, it follows that no $C_j$ can occur twice in the list $C_{k_1}, C_{k_2}, \dots$. Hence the algorithm must exit the loop after at most $k'$ iterations. Therefore, the computation for $\Mu{x}{t'}$ at $\vec{r}$ terminates. 
% , producing a conditioned linear expression $\GLE{C_{\vec{r}}\,}{\,e_{\vec{r}}}$.

It remains to show that the algorithm for $\Mu{x}{t'}$  produces only finitely many conditioned linear expressions $\GLE{C_{\vec{r}}\,}{\,e_{\vec{r}}}$. The crucial observation is that the vector $\vec{r}$ is used only to determine the control flow of the algorithm, i.e., which branches of conditional statements are followed, 
the choices made in selecting $N$ and $b_j$ in  (\ref{equation:next-iteration}), and the order 
in which the different $\GLE{C_{j}\,}{\,e_{j}}$, from
(\ref{equation:list}) are visited
(given by the sequence $k_1, k_2, \dots $ of values taken by $j$).
%In contrast, the data $(D_{i+1}, d_{i+1})$ 
%that is passed from one iteration to the next, is constructed combinatorially from $\GLE{C_{k_i}\,}{\,e_{k_i}}$ and $(D_i, d_i)$, in a way dependent on the control flow within the iteration of the loop, as is the eventual result $\GLE{C_{\vec{r}}\,}{\,e_{\vec{r}}}$, produced by the last iteration.
Using this, if $l'$ is the condition size of $t'$,
then a loose upper bound is that
the number of possible results $\GLE{C_{\vec{r}}\,}{\,e_{\vec{r}}}$ for the algorithm for $\Mu{x_{n+1}}{t'}$ is at most $(k'(l')^2)^{k'}$, and the number of inequalities in 
$C_{\vec{r}}$ is at most $2k'l'$.
\end{proof}

The above proof gives a truly abysmal complexity bound for the algorithm.
%(which would not be significantly improved by tightening the approximation in (\ref{equation:bound})).
%Let the basis size and condition size for terms $t_1,t_2$ be $k_1,k_2$ and $l_1,l_2$ respectively. Then, for example, the basis size $k$ and condition size $l$
%for $t_1 \strongor t_2$ are:
%\[
%k \: = \: 2 k_1 k_2 \qquad l \: = \: l_1 + l_2 + 1 \enspace .
%\]
Let the basis and condition size for the term $t'(x_1, \dots , x_{n+1})$ be
$k'$ and $l'$ respectively. 
Then, as in the proof, the basis and condition size for
$\Mu{x_{n+1}}{t'}$ are respectively bounded by:
\[
k \: \leq\: (k'(l')^2)^{k'} ~~ \text{and}~~  l \: \leq \: 2k'l' \enspace .
\]
Using these bounds, the basis and condition size have non-elementary growth in the number of fixed points in a term $t$.

\subsection{Comparison}\label{section:comparison}

According to the crude complexity analyses we have given, the evaluation of 
{\L}ukasiewicz $\mu$-terms via rational linear arithmetic is (in having doubly-\ and triply-exponential space and time complexity bounds) 
preferable to the (non-elementary space and hence time) evaluation via the direct algorithm. 
Nevertheless, we expect the direct algorithm to work better than this in practice. 
Indeed, a  main  motivating factor in the design of the direct algorithm 
is that the algorithm for $\Mu{x_{n+1}}{t'}$ 
only explores as much of the basis set for $t'$ as it needs to, and does so in an order that is tightly constrained by the monotone improvements made to the approximating $d$ expressions along the way. In contrast, the crude
complexity analysis is based on a worst-case scenario in which the
algorithm  is assumed to visit the entire basis for 
$t'$, and, moreover, to do so, for different input vectors $\vec{r}$, in every
possible order for  visiting the different basis sets. 
Perhaps better bounds can be obtained by a more careful analysis of the algorithm. 

\section{Model checking}
\label{section:model-checking}

Let $\phi$ be a closed $\Lukmu$ formula and $(S,\rightarrow)$ a finite rational PNTS. 
We wish to compute the value $\sem{\phi}(s)$ at any given state $s \in S$.
We do this by effectively producing a closed $\mu$-term $t_s(\phi)$, with the property that $t_s (\phi) = \sem{\phi}(s)$, whence the rational value of $\sem{\phi}(s)$ can be calculated by the algorithm in Section~\ref{section:algorithm}.

We assume, without loss of generality, that all fixed-point operators in $\phi$ bind distinct variables. Let $X_1, \dots, X_m$ be the variables appearing in $\phi$.
We write $\sigma_i \, X_i. \, \psi_i$ for the unique subformula of $\phi$ in which $X_i$ is bound. The strict (i.e., irreflexive) \emph{domination} relation $X_i \dominates X_j$ between variables is defined to mean that $\sigma_j \, X_j. \, \psi_j$ occurs as a subformula in $\psi_i$.

Suppose $|S| = n$. For each $s \in S$, we translate $\phi$ to a $\mu$-term $t_s(\phi)$ containing at most $mn$ variables $x_{i,s'}$, where $1 \leq i \leq m$ and $s' \in S$.
The translation is defined using a more general function
$t^\Gamma_s$, defined on subformulas of $\phi$, where 
$\Gamma \subseteq
\{1, \dots, m\} \times S$ is an  auxiliary component keeping track of the states at which variables have previously been encountered. Given $\Gamma$ and $(i,s) \in \{1, \dots, m\} \times S$, we define:
\[
\Gamma \dominates (i,s) \; = \; 
(\Gamma \cup \{(i,s)\}) \backslash
\{ (j,s') \in \Gamma \mid X_i \dominates X_j\} \enspace .
\]
This operation is used in the definition below
% Figure \ref{reduction_model_checking_figure} 
to `reset' subordinate fixed-point variables
whenever a new variable that dominates them is declared. 
%\begin{figure}
\begin{align*}
t^\Gamma_s (X_i) & =  
       \begin{cases} 
           x_{i,s} & \text{if $(i,s) \in \Gamma$} \\
           \sigma_i \,  x_{i,s} . \; t^{\Gamma \dominates (i,s)}_s(\psi_i) & \text{otherwise}
       \end{cases}
\\
t^\Gamma_s (P) & =  \underline{\rho(P)(s)}
\\
t^\Gamma_s (\negate{P}) & =  \underline{1-\rho(P)(s)}
\\
t^\Gamma_s (q \, \phi) & =  q \, t^\Gamma_s (\phi)
\\
t^\Gamma_s (\phi_1 \bullet \phi_2) & = 
     t^\Gamma_s (\phi_1) \bullet t^\Gamma_s (\phi_2) \qquad \bullet \in \{\weakor, \weakand, \strongor, \strongand\}
\\
t^\Gamma_s (\Diamond \phi) & = 
  \bigsqcup_{s \rightarrow d} ~ \bigoplus_{s' \in S} ~ d(s') \; t^\Gamma_{s'} (\phi) 
\\
t^\Gamma_s (\Box \phi) & = 
  \bigsqcap_{s \rightarrow d} ~ \bigoplus_{s' \in S} ~ d(s') \; t^\Gamma_{s'} (\phi) 
\\
t^\Gamma_s (\sigma_i \, X_i. \, \psi_i) & = 
  \sigma_i \,  x_{i,s} . \; t^{\Gamma \cup \{(i,s)\}}_s(\psi_i) 
\end{align*}
%\caption{Reduction of the Model-Checking problem to \L ukasiewicz $\mu$-terms evaluation.}
%\label{reduction_model_checking_figure}
%\end{figure}
\noindent
This is well defined because changing from
$\Gamma$ to $\Gamma \dominates (i,s)$ or to $\Gamma \cup \{(i,s)\}$ strictly increases the function 
\[ 
i \mapsto | \{ (i,s) \mid (i,s) \in \Gamma\}| \colon \{1, \dots, m\} \to \{0, \dots, n\}\] 
under the lexicographic order on functions relative to $\dominates$.

\begin{proposition}
For any closed  $\Lukmu$ formula $\phi$, finite PNTS $(S, \rightarrow)$ and $s \in S$, it holds that $\sem{\phi}(s) = t^\emptyset_s(\phi)$.
\end{proposition}
\noindent
We omit the laborious proof. It is reminiscent of 
the reduction of modal $\mu$-calculus model checking to a system of nested boolean fixed-point equations in Section 4 of~\cite{Mader1995}.

\section{Related and future work}

The first encodings of probabilistic temporal logics in a probabilistic version of the modal $\mu$-calculus were given in~\cite{CPN99}, where a version $\PCTLstar$, tailored to processes exhibiting probabilistic but not nondeterministic choice,  was translated into a 
non-quantitative probabilisitic variant of the $\mu$-calculus, which included 
explicit (probabilistic) path quantifiers but disallowed fixed-point alternation. 
%The temporal logic considered in~\cite{CPN99} was a  version $\PCTLstar$  restricted to processes exhibiting probabilistic choice only with no (non-probabilistic) nondeterminism.

In their original paper on quantitative  $\mu$-calculi~\cite{HM96}, Huth and Kwiatkowska attempted a model checking algorithm for alternation-free formulas in the version of $\Lukmu$ with $\strongor$ and $\strongand$ but without $\weakand$, $\weakor$ and
scalar multiplication. Subsequently, 
several authors have addressed the problem of 
computing (sometimes approximating) 
fixed points for monotone functions combining
linear (sometimes polynomial) expressions with $\min$ and $\max$ operations;
see \cite{GS2011} for a summary. However, such work has focused on (efficiently) finding outermost (simultaneous) fixed-points for systems of equations whose underlying monotone functions are continuous. The nested fixed points considered in the present paper give rise to the complication of non-continuous functions, as 
the example of Section \ref{section:example} demonstrates.

As future work, it is planned to run an experimental comparison of the direct algorithm against the reduction to  linear arithmetic. As suggested in Section \ref{section:comparison}, we expect the direct algorithm to work better in practice than the non-elementary upper bound on its complexity, given by our crude analysis, suggests. Furthermore, as a natural generalization of the approximation approach to computing fixed points, the direct algorithm should be amenable to optimizations such as the simultaneous solution of adjacent fixed points of the same kind, and the reuse of previous approximations when applicable due to  monotonicity considerations. Unlike the black-box reduction to linear arithmetic, based on quantifier elimination, the linear-constraint-based approach of the direct algorithm should also offer a flexible machinery helpful in the design of optimized procedures for calculating values of particular subclasses of $\Lukmu$-terms.
%(i.e., sub-logics) expressing simpler constraints. 
An important example is given by the fragment of $\Lukmu$ capable of encoding $\PCTL$ (see Remark \ref{remark_fragment_plmu}).

Our results on $\Lukmu$ are a contribution towards the development of a robust theory of fixed-point probabilistic logics. The simplicity of the proposed encoding of $\PCTL$ (see Remark \ref{remark_fragment_plmu} above) suggests that the direction we are following is promising. In a follow-up paper, by the first author, it will be shown that the process equivalence characterised by {\L}ukasiewicz $\mu$-calculus is the standard notion of \emph{probabilistic bisimilarity}~\cite{S95}. Thus the quantitative approach to probabilistic $\mu$-calculi may be considered equally suitable as a mechanism for 
characterising process equivalence as the  non-quantitative $\mu$-calculi
advocated for this purpose in~\cite{CPN99} and~\cite{DvG2010}.

Further research will have to explore the relations between quantitative $\mu$-calculi such as $\Lukmu$ and other established frameworks for verification and design of probabilistic systems. Important examples include the \emph{abstract probabilistic automata} of~\cite{DKLLPSW}, the compositional \emph{assume-guarantee} techniques of~\cite{KNPQ2010,FKNPQ2011} and the recent \emph{p-automata}  of~\cite{pautomata2012}. In particular, with respect to the latter formalism, we note that the acceptance condition of p-automata is specified in terms of stochastic games whose configurations may have preseeded threshold values whose action closely resembles that of the threshold modalities considered in this work (Definition \ref{thresholds_encoding}). Exploring the relations between p-automata games and $\Lukmu$-games~\cite{MioThesis} could shed light on some underlying fundamental ideas.

\section*{Acknowledgements}
We thank Kousha Etessami, Grant Passmore, Colin Stirling and the anonymous reviewers for helpful comments  and for pointers to the literature.

The first author carried out this work during the tenure of an ERCIM ``Alain Bensoussan'' Fellowship, supported by the Marie Curie Co-funding of Regional, National and International Programmes (COFUND) of the European Commission. 

\bibliographystyle{eptcs}

%\bibliography{biblio,alex}

\appendix

\section{Appendix: some omitted proof details}

We add detail to the outlined proof of Theorem~\ref{theorem:pctl}, by supplying the omited argument for the equality 
\[
\bigsqcup_\sigma \{ m^s_\sigma(\Psi)\} = \sem{\mu X.F(X)\big) }_\rho(s) \enspace ,
\]
which appears as case \ref{case_threshold_3}. Although game semantics provides the most intuitive justification, we instead give a direct denotational proof, in order to avoid introducing game-theoretic machinery.

\begin{proof}[Expanded proof of Theorem~\ref{theorem:pctl}]
{Case \ref{case_threshold_3}} $\mathbf{(\leq)}$. We first show that 
\begin{equation}\label{proof_eq_1}
 \bigsqcup_\sigma \{ m^s_\sigma(\Psi)\} \leq \sem{\mu X.F(X)\big) }_\rho(s)
\end{equation}
Define $\Psi_k=\{ {s_0.s_1.s_2\dots} \ | \ s_0=s \textnormal{ and }\exists n\leq k . \big( s_n\in \gsem{\phi_2}_\rho \wedge \forall m<n.( s_m\in\gsem{\phi_1}_\rho ) \big)\}$. Clearly $\Psi=\bigcup_k \Psi_k$. Suppose Inequality \ref{proof_eq_1} does not hold. Then there exists some $k$ and scheduler $\sigma$ such that 
\begin{equation}\label{proof_eq_2}
m^s_\sigma(\Psi_k)> \sem{\mu X. F(X) }_\rho(s)
\end{equation}
We prove that this is not possible by induction on $k$.  In the $k=0$ case, since we are assuming $m^s_\sigma(\Psi_0)>0$, it holds that $s\in\gsem{\phi_2}_\rho$. By inductive hypothesis on $\phi_2$, we know that $\sem{\mathbf{E}(\phi_2)}(s)=1$ and this implies that $\mu X.F(X)=1$, which is a contradiction with the assumed strict inequality \ref{proof_eq_2}.  Consider the case $k+1$. Note that if $s\in\gsem{\phi_2}_\rho$ then, $\sem{\mu X.F(X) }_\rho(s)=1$ as before, contradicting Inequality \ref{proof_eq_2}. So assume $s\not\in\gsem{\phi_2}_\rho$. Since we are assuming $m^s_\sigma(\Psi_{k+1})>0$ it  must be the case that $s\in\gsem{\phi_1}_\rho$. Similarly, $m^s_\sigma(\Psi_{k+1})>0$ and $s\not\in\gsem{\phi_2}_\rho$ imply that $s\not\rightarrow$ does not hold. This means (see Definition \ref{scheduler_def}) that $\sigma(\{s\})$ is defined. Let $d=\sigma(\{s\})$ and observe that $m^s_\sigma(\Psi_{k+1})=\displaystyle \sum_{t\in S}d(t)m^t_{\sigma^\prime}(\Psi_k)$, where $\sigma^\prime(s_0,s_1,\dots,s_n)=\sigma(s,s_0,s_1,\dots,s_n)$.  By induction on $k$ we know that the inequality $m^t_{\sigma^\prime}(\Psi_k)\leq \sem{\mu X.F(X) }_\rho(t)$ holds for every $t \in S$. Thus, by definition of the semantics of $\Diamond$, we obtain
$ m^s_\sigma( \Psi_k) \leq \sem{\Diamond \big( \mu X.F(X) \big) }_\rho$. 
Recall that we previously assumed $s\not\in\gsem{\phi_2}_\rho$ and $s\in\gsem{\phi_1}_\rho$. Hence the equality
$$\sem{\Diamond \big( \mu X.F(X) \big) }_\rho(s) = \sem{\mathbf{E}(\phi_2) \sqcup ( \mathbf{E}(\phi_1) \sqcap \big( \Diamond \mu X.F(X) \big) )  }_{\rho}(s)$$
holds. The formula on the right is just the unfolding $F(\mu X.F(X))$ of $\mu X.F(X)$. This implies the desired contradiction. 

{Case \ref{case_threshold_3}}$\mathbf{(\geq)}$. We now prove that also the  inequality
\begin{equation}\label{proof_eq_3}
 \bigsqcup_\sigma \{ m^s_\sigma(\psi)\} \geq \sem{\mu X. F(X) }_\rho(s)
\end{equation}
holds.  By Knaster-Tarski theorem, $\sem{\mu X.F(X)}_\rho=\bigsqcup_\alpha \sem{F(X)}^{\alpha}_\rho$, where $\alpha$ ranges over the ordinals and $\sem{F(X)}^{\alpha}_{\rho^\alpha}$ with $\rho^{\alpha}=\rho[\bigsqcup_{\beta<\alpha} \sem{F(X)}_{\rho^\beta}/ X]$. We prove Inequality \ref{proof_eq_3} by showing, by transfinite induction, that for every ordinal $\alpha$ and $\epsilon >0$,  the inequality 
\begin{equation}
\bigsqcup_\sigma \{ m^s_\sigma(\psi)\} > \sem{\mu X. F(X) }_{\rho^\alpha}(s)-\epsilon
\end{equation} 
holds, for all $s\in S$.  The case for $\alpha=0$ is immediate since $\sem{F}_{\rho^0}(s)>0$ if and only if $\sem{\mathbf{E}(\phi_2)}_\rho(s)=1$ and this implies  $\bigsqcup_\sigma \{ m^s_\sigma(\psi)\} =1$. Consider $\alpha=\beta+1$. If $\sem{\mathbf{E}(\phi_2)}_\rho(s)=1$ then Inequality \ref{proof_eq_3} holds as above. Thus assume $\sem{\phi_2}_\rho(s)=0$. Note that $\sem{F}_{\rho^0}(s)>0$ only if $s\in \sem{\mathbf{E}(\phi_1)}$. Thus assume $\sem{\mathbf{E}(\phi_1)}_\rho^\beta (s)=1$.
Under these assumption,  $\sem{F(X)}_{\rho^{\alpha}}= \sem{\Diamond F(X)}_{\rho^\beta}$ as it is immediate to verify. By definition of the semantics of $\Diamond$ we have: 
$$
\sem{\Diamond F(X)}_{\rho^\beta}(s)= \bigsqcup_{s\rightarrow d} \big( \displaystyle \sum_{t\in S} d(t)\sem{F(X)}_{\rho^\beta}(t)\big)
$$
By induction hypothesis on $\beta$ we know that for every $\epsilon$, 
$$
\sem{\Diamond F(X)}_{\rho^\beta}(s)  < \bigsqcup_{s\rightarrow d} \big( \displaystyle \sum_{t\in S} d(t)\Big( \bigsqcup_\sigma \{ m^t_\sigma(\psi)\} + \epsilon \Big) \big)
$$
For each $s\rightarrow d$ and $\sigma$ define $\sigma^d$ as $\sigma^d(\{s\})=d$ and $\sigma^d(s.t_0.\dots)=\sigma(t_0\dots)$. A simple argument shows that 
$$
\bigsqcup_{s\rightarrow d} \big( \displaystyle \sum_{t\in S} d(t)\Big( \bigsqcup_\sigma \{ m^t_\sigma(\psi)\} + \epsilon \Big) \big) =  \bigsqcup_{\sigma^d} \{ m^s_{\sigma^{d}}(\psi)\} + \epsilon
$$
and this conclude the proof for the case $\alpha=\beta+1$. Lastly, the case for $\alpha$ a limit ordinal follows straightforwardly from the inductive hypothesis on $\beta<\alpha$.
\end{proof}

\begin{proof}[Proof of Proposition~\ref{proposition:cle:arithmetic}]
Suppose we have a system of $k$ conditioned linear expressions representing $f$.
Each conditioned expression $\GLE{C\,}{\,e}$ is captured by the implication $(\bigwedge C) \rightarrow y = e$, so the whole system translates into a conjunction of $k$ such implications. To this conjunction, one need only add the range constraints $0 \leq 
z$ and $z \leq 1$ for each variable $z$, as further conjuncts. In this way, the graph is easily expressed as a quantifier free formula. (Since the implications are equivalent to disjunctions of atomic formulas, the resulting formula is naturally in conjunctive normal form.) 

Conversely, suppose $F(x_1, \dots, x_n, y)$ defines the graph of $f$. By  quantifier elimination, we can assume that $F$ is quantifier free and in disjunctive normal form. Then $F$ is a disjunction of conjunctions, where each conjunction, $K$, can be easily rewritten in the form
\begin{equation}
\label{equation:K}
\left(\bigwedge C  \right)
\, \wedge \, 
\left(\bigwedge_{1 \leq i \leq h} y > a_i \right)
\, \wedge \, 
\left(\bigwedge_{1 \leq i \leq k} y \geq b_i\right)
\, \wedge \, 
\left(\bigwedge_{1 \leq i \leq l} y \leq c_i\right)
\, \wedge \, 
\left(\bigwedge_{1 \leq i \leq m} y < d_i\right) \enspace ,
\end{equation}
such that the only variables in the finite set of atomic formulas $C$, and linear expressions $a_i, b_i, c_i, d_i$ are $x_1, \dots, x_n$. 
Since $F$ is the graph of a function, for all reals
$r_1, \dots, r_n$, there is at most one $s$ such that 
$K(\vec{r},s)$ holds, and, if it does, then all of $r_1, \dots, r_n, s$ are in $[0,1]$.
Given such an $s$, we therefore  have:
\[
\max\{a_i(\vec{r}) \mid 1 \leq i \leq h\} 
< 
\max\{b_i(\vec{r}) \mid 1 \leq i \leq k\} = s = \min\{c_i(\vec{r}) \mid 1 \leq i \leq l\} < \min\{d_i(\vec{r}) \mid 1 \leq i \leq m\} \enspace .
\]
A system of conditioned linear expressions for $f$ is thus obtained as follows.
For each conjunct $K$ in $F$, written in the form of (\ref{equation:K}) above,
and each $j$ with $1 \leq j \leq k$, include the conditioned linear expression:
\[
\GLE{
C, \,  
\{b_j > a_i \}_{1 \leq i \leq h}, \, 
\{b_j \geq b_i\}_{1 \leq i \leq k}, \, 
 \{b_j \leq c_i\}_{1 \leq i \leq l}, \,
 \{b_j <  d_i\}_{1 \leq i \leq m}, \,
}{\, b_j} \enspace .
\]
\end{proof}

We supplement the proof of Theorem~\ref{theorem:algorithm-correct} with more detail on the bounds on basis and condition size.

\begin{proof}[Expanded proof of Theorem~\ref{theorem:algorithm-correct}]
We analyse the control flow in the algorithm for 
$\Mu{x_{n+1}}{t'}$  on a given input vector $(r_1, \dots, r_n)$. On iteration number $i$, the loop is 
entered with constraints $D_i$ and approximation $d_i$, after which 
the recursive call to the 
algorithm for $t'$ yields one
of the conditioned linear expressions, $\GLE{C_{k_i}\,}{\,e_{k_i}}$.
Suppose that $C_{k_i}$ and $D_i$ contain $u$ and $v$ inequalities respectively. If the loop is exited producing (\ref{equation:first-result}) as result then the resulting $C_{\vec{r}}$ has $2u +v$ inequalities. 
If it is exited producing (\ref{equation:second-result}) as result then 
$C_{\vec{r}}$ has $u +v+2$ inequalities (where $u + v + 2 \leq 2u+v$ because $C_{k_i}$ has to enforce the range constraint $0 \leq x_{n+1} \leq 1$). Otherwise, the algorithm repeats the loop, entering iteration $i+1$ with $D_{i+1}$, given by~(\ref{equation:next-iteration}),
having at most $2u + v$ inequalities ($N$ contributes $1$ inequality, and
there are at most $u-1$ inequalities  $b_j \leq b_i$ in (\ref{equation:next-iteration}) since $l \geq 1$). 

Therefore, if $l'$ is now 
maximum number of inequalities occurring in any
$C_j$ from 
(\ref{equation:list}) (i.e., if it is the condition size for $t'$)
the algorithm for $\Mu{x_{n+1}}{t'}$ at $\vec{r}$, which runs for at most $k'$ iterations, results in $C_{\vec{r}}$ containing at most $2k'l'$ inequalities. 

To bound the number of  results $\GLE{C_{\vec{r}}}{e_{\vec{r}}}$, we count the possible control flows of the algorithm. At iteration $i$, the algorithm uses $\GLE{C_{k_i}\,}{\,e_{k_i}}$ from
(\ref{equation:list}), using which it might terminate with either 
(\ref{equation:first-result}) or (\ref{equation:second-result}), or it might
repeat the loop, 
entering iteration $i+1$ with $D_{i+1}$, given by~(\ref{equation:next-iteration}),
which can arise from $C_{k_i\,}$ in a number of ways determined by the possible
pairs of choices for $N$ and $b_j$ in (\ref{equation:next-iteration}).
In the case that the variable vector $(x_1, \dots, x_n)$ is empty (i.e., the term
$\Mu{x_{n+1}}{t'}$ is closed) the constraints in $D$ are redundant (they are simply true inequalities between rathionals) and so can be discarded. 
In the case that $n \geq 1$, there are at least $2$ inequalities in $C$ giving range constraints on $x_1$, so there are at most $l'$ choices for $N$
($l'-2$ choices in the case that $q_{n+1} \neq 1$, and $2$ in the case $q_{n+1}= 1$).
Irrespective of $n$, there are at most $l'-1$ choices for $b_j$ (taking $n$ into account this can be improved to $l'-2n-1$).
Therefore, the execution
of the algorithm, is determined by the sequence:
\[
k_1, \,  u_1, \, k_2, \, u_2, \, \dots, \, k_m , \, v
\]
where: $m \leq k'$ is the number of loop iterations performed; each $u_i$, where  $1 \leq u_i \leq l'(l'-1)$, represents the choice of $N$ and $b_j$ used in the construction of $D_{i+1}$ (\ref{equation:next-iteration}), and $v$ is $ 1$ or $2$ according to whether 
the resulting $\GLE{C_{\vec{r}}\,}{\,e_{\vec{r}}}$ is  returned 
via (\ref{equation:first-result}) or (\ref{equation:second-result}). Since each
number $k_i$ is distinct, the number of different such sequences is bounded by:
\begin{equation}
\label{equation:bound}
2 \sum_{m = 1}^{k'} \frac{k'!}{(k'-m)!} (l'\,(l'-1))^{m-1} \; \leq \; (k'(l')^2)^{k'} \enspace ,
\end{equation}
where the right-hand-side gives a somewhat loose upper bound. Therefore, the number of possible results $\GLE{C_{\vec{r}}\,}{\,e_{\vec{r}}}$ for the algorithm for $\Mu{x_{n+1}}{t'}$ is at most $(k'(l')^2)^{k'}$.
\end{proof}

\end{document}